\documentclass[12pt]{amsart}
\usepackage{amsmath}
\usepackage{latexsym}
\usepackage{amsfonts}
\usepackage{amssymb}
\usepackage{bbm,dsfont}
\usepackage{color}

\newtheorem{proposition}{Proposition}
\newtheorem{theorem}{Theorem}

\theoremstyle{definition}

\newtheorem{remark}{Remark}
\newtheorem{example}{Example}
\newtheorem{definition}{Definition}

\newcommand{\R}{\mathbb R} 
\newcommand{\C}{\mathbb C} 
\newcommand{\N}{\mathbb N} 
\newcommand{\Z}{\mathbb Z} 
\newcommand{\T}{\mathbb T} 
\newcommand{\abs}[1]{\left| #1 \right|} 

\newcommand{\hi}{\mathcal{H}} 
\newcommand{\hh}{\hi} 
\newcommand{\hik}{\mathcal{K}} 
\newcommand{\kk}{\hik} 
\newcommand{\lh}{\mathcal{L(H)}} 
\newcommand{\trh}{\mathcal{T(H)}} 
\newcommand{\sh}{\mathcal{S(H)}} 
\newcommand{\shd}{\mathcal{S(H)_{\rm diag}}} 
\newcommand{\ch}{{\rm Ch} (\mathcal{H})}
\newcommand{\CPH}{\ch^*} 
\newcommand{\uh}{\mathcal{U(H)}} 
\newcommand{\ip}[2]{\left\langle\,#1\,|\,#2\,\right\rangle} 
\newcommand{\ket}[1]{|#1\rangle} 
\newcommand{\kb}[2]{|#1\,\rangle\langle\,#2|} 

\newcommand{\nil}{O} 

\newcommand{\no}[1]{\left\|#1\right\|} 
\newcommand{\tr}[1]{\mathrm{tr}\left[#1\right]} 

\newcommand{\group}{G} 

\newcommand{\obs}{\mathcal{O}} 
\newcommand{\E}{\mathsf{E}} 
\newcommand{\F}{\mathsf{F}} 
\newcommand{\Enu}{\E_\nu} 
\newcommand{\Ecan}{\E_{\rm can}} 
\newcommand{\Es}{\E_{\ket{s}}} 
\newcommand{\CC}{\mathcal C} 

\newcommand{\pre}{\textsc{pre}}
\newcommand{\post}{\textsc{post}}
\newcommand{\csim}{\stackrel{\post}{\sim}}  
\newcommand{\coar}{\preccurlyeq_{\post}}

\newcommand{\dsim}{\stackrel{\pre}{\sim}}  
\newcommand{\dis}{\preccurlyeq_{\pre}}

\newcommand{\bor}[1]{\mathcal{B}(#1)} 
\newcommand{\borG}{\bor{\group}} 
\newcommand{\borO}{\mathcal{B}(\Omega)} 
\newcommand{\borT}{\bor{\T}} 
\newcommand{\prob}[1]{M_1^+(#1)} 
\newcommand{\probO}{\prob{\Omega}} 
\newcommand{\probG}{\prob{\group}} 
\newcommand{\linf}{L^\infty (\T)}
\newcommand{\ldue}[1]{L^2 \left( #1 \right)}

\newcommand{\de}{\, d}

\newcommand{\frecc}{\rightarrow}

\newcommand{\lft}{\left(}
\newcommand{\rgt}{\right)}
\newcommand{\ext}{\mathop{\rm Ext}}
\newcommand{\rank}{\mathop{\rm rank}}

\begin{document}

\title[Optimal phase observables]{Optimal covariant measurements: the case of a compact symmetry group and phase observables}

\author[Carmeli]{Claudio Carmeli}
\address{Claudio Carmeli, Dipartimento di Fisica, Universit\`a di Genova, and I.N.F.N., Sezione di Genova, Via Dodecaneso 33, 16146 Genova, Italy}
\email{carmeli@ge.infn.it}

\author[Heinosaari]{Teiko Heinosaari}
\address{Teiko Heinosaari, Department of Physics, University of Turku, 
FIN-20014 Turku, Finland}
\email{heinosaari@gmail.com}

\author[Pellonp\"a\"a]{Juha-Pekka Pellonp\"a\"a}
\address{Juha-Pekka Pellonp\"a\"a, Department of Physics, University of Turku, 
FIN-20014 Turku, Finland}
\email{juhpello@utu.fi}

\author[Toigo]{Alessandro Toigo}
\address{Alessandro Toigo, Dipartimento di Informatica, Universit\`a di Genova, Via Dodecaneso 35, and I.N.F.N., Sezione di Genova, Via Dodecaneso 33, 16146 Genova, Italy}
\email{toigo@ge.infn.it}

\date{\the\day/\the\month/\the\year}

\maketitle

\begin{abstract}
We study various optimality criteria for quantum observables. Observables are represented as covariant positive operator valued measures and we consider the case when the symmetry group is compact. Phase observables are examined as an example.
\end{abstract}


\maketitle

\section{Introduction}\label{sec:intro}

In this work we study four optimality criteria applied on covariant quantum observables. Quantum observables are described as normalized positive operator valued measures. Covariance means that observables transform in a consistent way under a group operation. For example, phase observables are defined as the normalized positive operator measures on a circle which are covariant under the phase shifts.  

A quantum observable can be optimal in various ways. There are at least four different ways how optimality can be defined. Namely, an observable can be
\begin{itemize}
\item[(a)] (approximately) sharp
\item[(b)] extremal 
\item[(c)] postprocessing clean
\item[(d)] preprocessing clean
\end{itemize}
If an observable is (approximately) sharp then for any neighborhood of a point from the set of possible measurement outcomes, one can prepare a state such that the probability of getting a result from the neighborhood is (approximately) 1. This reflects the fact that an approximately sharp observable does not have intrinsic unsharpness -- spread in the measurement outcome distribution depends solely on the measured state. 

A covariant observable is extremal in the convex set of all covariant observables if it cannot be represented as a nontrivial convex combination of other covariant observables. A convex mixture of observables corresponds to a random choice between measurement apparatuses, and thus an extremal observable is unaffected by this kind of randomness.

After a measurement of an observable, one may try to process the obtained measurement outcome data to get the measurement outcome data of another observable. An observable is called postprocessing clean if it cannot be obtained by manipulating the measurement outcome data of another observable.

Analogously, a preprocessing clean observable cannot be obtained by manipulating the state before the measurement and then measuring some other observable. State manipulation is described by a quantum channel, which typically loses some information. Hence, a preprocessing clean observable is not irreversibly connected to another observable by a channel.

In Section \ref{sec:covariant} we formulate the optimality criteria and derive some general results in the case of a compact symmetry group. In Section \ref{sec:phase} we focus on the phase observables. The phase observables form an interesting class as there is no sharp phase observable (i.e.\ projection valued measure). This class is also a rich example for the general investigation on optimal observables, showing various connections and differences between the four optimality criteria.

\section{Optimal covariant observables}\label{sec:covariant}

\subsection{Basic definitions}\label{sec:definitions}

In quantum mechanics, observables are represented as normalized positive operator valued measures. We shortly recall some relevant basic concepts. For more details, we refer to \cite{OQP97}, \cite{FQMEA02}, \cite{SSQT01}.

Let $\hi$ be a complex separable Hilbert space. We denote by $\lh$ the set of bounded operators and $\trh$ the set of trace-class operators on $\hi$. We denote by $O$ and $I$ be the zero and indentity operators on $\hi$, respectively. Let $\Omega$ be a topological space. We denote by $\borO$ the Borel $\sigma$-algebra on $\Omega$.

\begin{definition}\label{def:observable}
A set function $\E:\borO\to\lh$ is an \emph{observable} if it satisfies the following conditions:
\begin{itemize}
\item $O\leq \E(X) \leq I$ for any $X\in\borO$;
\item $\E(\Omega)=I$;
\item $\E(\cup_i X_i)=\sum_i \E(X_i)$ for any sequence $\{X_i\}$ of disjoint Borel sets, the sum converging in the weak operator topology.
\end{itemize} 
The set $\Omega$ is called the \emph{outcome set} of $\E$.

Two observables $\E$ and $\F$ are called {\em equivalent}, and denoted $\E\sim\F$, if there exists a unitary operator $W\in\lh$ such that $\F(X)=W\E(X)W^\ast$ for all $X\in\borO$.
\end{definition}

States of a quantum system are described by (and identified with) positive trace class operators of trace 1. We denote by $\sh$ the convex set of all states.  If $\E$ is an observable and $\varrho$ is a state, then the trace formula 
\begin{equation*}
p^\E_\varrho(X):=\tr{\varrho\E(X)},\quad X\in\borO,
\end{equation*}
defines a probability measure $p^\E_\varrho$ on $\borO$. The number $p^{\E}_{\varrho}(X)$ is interpreted as the probability of getting a measurement outcome $x$ belonging to $X$, when the system is in the state $\varrho$ and the observable $\E$ is measured. 

There is also an equivalent description of observables which we need later. Namely, an observable $\E$ determines a mapping $\Theta_\E$ from the set of states $\sh$ into the set of probability measures $\probO$,
\begin{equation*}
\varrho\mapsto\Theta_\E(\varrho):=p^\E_\varrho\, .
\end{equation*}
The mapping $\Theta_\E$ is affine, i.e., it maps convex combinations of states to convex combinations of corresponding probability measures. 

In this paper we study observables which have a specific symmetry property, called covariance. To formulate this concept, let $\group$ be a topological group. For simplicity, here we consider only the case where $\group$ is the symmetry group and also the outcome space. (This corresponds to the situation when the symmetry group acts transitively on the outcome space $\Omega$ and the stability subgroup is the trivial one element group consisting of the identity element $e$ only.) By a unitary representation of $\group$ we mean a strongly continuous group homomorphism from $\group$ to the group $\uh$ of unitary operators in a Hilbert space $\hi$.

\begin{definition}
Let $U$ be a unitary representation of $\group$ in a Hilbert space $\hi$. An observable $\E:\borG \frecc \lh$ is {\em covariant} with respect to $U$ (or {\em $U$-covariant}, for short) if 
\begin{equation*}
U(g) \E(X) U(g)^\ast = \E(gX)
\end{equation*}
for all $g\in G$ and $X\in\borG$.

Two $U$-covariant observables $\E$ and $\F$ are called {\em $U$-equivalent}, and denoted $\E\sim_U\F$, if there exists a unitary operator $W\in\lh$ such that $U(g)W=WU(g)$ for all $g\in G$ and $\F(X)=W\E(X)W^\ast$ for all $X\in\borG$.
\end{definition}

Covariance of observables arises in many different situations; we refer to \cite{OQP97}, \cite{PSAQT82} for many examples. Trivially, for two $U$-covariant observables $\E$ and $\F$, the equivalence $\E\sim_U\F$ implies that $\E\sim\F$ but, as we will show later (see Proposition \ref{prop:not-U-equivalent}), $\E\sim\F$ does {\it not} necessarily imply that $\E\sim_U\F$.

In the rest of this paper, $\group$ is a compact topological group which is Hausdorff and satisfies the second axiom of countability. We denote by $\mu$ the Haar measure of $\group$, normalized so that $\mu(\group)=1$. We fix a unitary representation $U$ of $\group$, and the set of all $U$-covariant observables is denoted by $\obs$.

\subsection{Approximately sharp observables}\label{sec:sharp}

Suppose for a moment that $\E:\borG \frecc \lh$ is a projection valued observable, also called a {\it sharp observable}.  Then for any $X\in\borG$ with $\E(X)\ne O$, we can choose a (pure) state $\varrho\in\sh$ such that $\tr{\varrho\E(X)}=1$. This means that we can prepare the system in a state $\varrho$ such that the probability of getting an outcome $x\in X$ is 1. 

This property of sharp observables is a simple consequence of the fact that each non-zero operator $\E(X)$ has eigenvalue 1. Especially, $\no{\E(X)}=1$. This gives a useful generalization of sharp observables  \cite{HeLaPePuYl03},\cite{SSQT01}.

\begin{definition}\label{def:as}
An observable $\E:\borG \frecc \lh$ is \emph{approximately sharp} if $\no{\E(X)}=1$ for every open set $X\subseteq\group$ such that $\E(X)\neq O$.
\end{definition}

Let us first note that the condition $\no{\E(X)}=1$ does not imply that $\tr{\varrho\E(X)}=1$ for some state $\varrho$. Indeed, $\no{\E(X)}=1$ means that 1 is in the spectrum of $\E(X)$ but 1 need not be an eigenvalue. The condition $\no{\E(X)}=1$ is equivalent to the fact that there exists a sequence $\{\varrho_n\}_{n\in\N}$ of (pure) states such that
\begin{equation*}
\lim_{n\to\infty} \tr{\varrho_n \E(X)}=1 \, .
\end{equation*}
Therefore, approximate sharpness means that we can prepare the system in a state $\varrho_{n}$ such that the probability of getting an outcome from $X$ is arbitrarily close to 1. 

Now, assume that $\E$ is a $U$-covariant observable. Then $\E(X)=\nil$ if and only if $\mu(X)=0$; see \cite{HeLaPePuYl03}. On the other hand, $\mu(X)>0$ for any (non-empty) open set $X\subseteq\group$. Thus, $\E$ is approximately sharp if $\no{\E(X)}=1$ for every open set $X\subseteq\group$, $X\ne\emptyset$.

\begin{proposition}
Let $\E,\F\in\obs$ and $\E\sim\F$. If $\E$ is approximately sharp, then also $\F$ is such. 
\end{proposition}

\begin{proof}
Since $\E\sim\F$, there is a unitary operator $W$ such that $\F(X)=W\E(X)W^\ast$ for every $X\in\borG$. This implies that the operator norms of $\E(X)$ and $\F(X)$ are the same, i.e., $\no{\F(X)}=\no{W\E(X)W^*}=\no{\E(X)}$.
\end{proof}

Until now, it may seem that the concept of an approximately sharp observable is quite artificial. The usefulness and importance of this concept becomes evident in the situations when there are no sharp observables in the set $\obs$, but there exist approximately sharp observables in $\obs$. This is the case  for phase observables, which we will study in Section \ref{sec:phase}.

\subsection{Extremal observables}\label{sec:convex}

The set of all $U$-covariant observables $\obs$ is convex, and we denote by $\ext\obs$ the set of extremal elements of $\obs$. A convex mixture of two observables corresponds to a random choice or fluctuation between two measurement apparatuses. An extremal observable thus describes an observable which is unaffected by this kind of randomness.

In our earlier article \cite{CaHePeTo08} we have characterized the extremal observables in $\obs$. We do not reproduce this characterization here but we use it in the special case of phase observables in Section \ref{sec:phase-convex}. Extremality and approximate sharpness have the following simple connection.

\begin{proposition}\label{prop:convex-as}
Let $\E,\E_1,\E_2$ be three observables and assume that $\E$ is a convex combination of $\E_1$ and $\E_2$. If $\E$ is approximately sharp, then also $\E_1$ and $\E_2$ are approximately sharp.
\end{proposition}

\begin{proof}
Let $X\subseteq\group$ be an open set. As $\E$ is a convex combination of $\E_1$ and $\E_2$, we have $\E(X)=\lambda \E_1(X) + (1-\lambda) \E_2(X)$ for some $0<\lambda<1$. By the triangle inequality, we get
\begin{equation*}
\no{\E(X)} \le \lambda \no{\E_1(X)} + (1-\lambda) \no{\E_2(X)} \leq 1 \, .
\end{equation*}
Thus, if $\no{\E(X)}=1$, then $\no{\E_1(X)}=\no{\E_2(X)}=1$.
\end{proof}

\subsection{Postprocessing}\label{sec:post}

Let $\probG$ be the convex set of probability measures on $\borG$. As we earlier discussed, each observable $\E$ defined on $\borG$ determines an affine mapping $\Theta_\E$ from $\sh$ to $\probG$.

\begin{definition}\label{def:post}
Let $\E,\F \in\obs$. We say that $\F$ is a \emph{postprocessing} of
$\E$, and denote $\F\coar\E$, if there exists an affine mapping
$\Psi:M^+_1(\group) \to M^+_1(\group)$ such that 
\begin{equation}\label{eq:fpsie}
\Theta_\F=\Psi \circ \Theta_\E \, .
\end{equation}
\end{definition}

The relation $\coar$ on $\obs$ is reflexive and transitive, thus a preorder. We denote by $\csim$ the induced equivalence relation, i.e., $\E\csim\F$ if and only if $\F\coar\E \coar\F$.

\begin{definition}\label{def:post-clean}
An observable $\E\in\obs$ is \emph{postprocessing clean in $\obs$} if for every $\F\in\obs$, the following implication holds: 
\begin{equation}\label{eq:post-clean}
\E\coar\F\ \Rightarrow\ \E\csim\F \, .
\end{equation}
\end{definition}

The postprocessing relation and postprocessing cleaness have been studied in \cite{BuDaKePeWe05}, \cite{FQMEA02}, \cite{Heinonen05}, \cite{JePuVi08}, \cite{MaMu90a}. Some connections with the convex structure of the set of observables and the postprocessing relation were proved in \cite{JePu07}. In the following we derive some results for the postprocessing relation in the case of covariant observables.

For each probability measure $\nu$ on $\group$, we define a mapping $\Psi_\nu:\probG\to\probG$ by
\begin{equation}
\Psi_{\nu}(p)=\nu\ast p \, .
\end{equation}
Here $\nu\ast p$ denotes the convolution of these measures, i.e.,
\begin{equation*}
(\nu\ast p)(X) = \int \nu(g^{-1}X)\ d p(g)\, ,\quad X\in\borG \, .
\end{equation*}
For $\E\in\obs$ and $\nu\in\probG$, we denote by $\Enu$ the observable corresponding to the affine mapping $\Psi_\nu\circ\Theta_\E$. In other words, $\Enu$ is the observable defined by formula
\begin{equation}\label{eq:Enu}
\Enu(X)=\int \nu(g^{-1}X)\ d \E(g) \, ,\quad X\in\borG \, .
\end{equation}
By its definition, $\Enu$ is a postprocessing of $\E$. It is straightforward to check that the $U$-covariance of $\E$ implies that also $\Enu$ is $U$-covariant.

\begin{proposition}\label{prop:post}
Let $\E,\F\in\obs$. The following conditions are equivalent:
\begin{itemize}
\item[(i)] $\F\coar\E$ .
\item[(ii)] There is a probability measure $\nu$ such that $\F=\Enu$ .
\end{itemize}
\end{proposition}  

\begin{proof}
As explained previously, (ii) implies (i). 

Conversely, suppose
$\Psi:M^+_1(\group) \to M^+_1(\group)$ is an affine mapping such that Eq.~\eqref{eq:fpsie} holds. The mapping $\Psi$ extends uniquely to a bounded linear mapping $\Psi:M(\group) \to M(\group)$, where $M(\group)$ is the Banach space of Borel complex measures on $G$,
 by setting
\begin{eqnarray*}
\Psi (m) &=& \no{m_{1+}} \Psi \left( \frac{m_{1+}}{\no{m_{1+}}} \right) - \no{m_{1-}} \Psi \left( \frac{m_{1-}}{\no{m_{1-}}} \right) \\
&& + i \no{m_{2+}} \Psi \left( \frac{m_{2+}}{\no{m_{2+}}} \right) - i \no{m_{2-}} \Psi \left( \frac{m_{2-}}{\no{m_{2-}}} \right) \, .
\end{eqnarray*}
Here $m = m_{1+} - m_{1-} + im_{2+} - im_{2-}$ is the Lebesgue decomposition of a measure $m$ and $\| m \|$ is the total variation norm of $m\in M(\group)$
(in the above formula, $0/0 = 0$ is assumed).

For each $g\in\group$, we define the following linear isometric isomorphism $\Lambda_g: M(G) \to M(G)$
\begin{equation*}
\int_G f(h) \de (\Lambda_g m) (h) := \int f(gh) \de m (h) \quad \forall f\in C_0 (G) \, .
\end{equation*}
The mapping $g\mapsto \Lambda_g m$ is weak*-continuous, and it is norm continuous if $m \in L^1 (G)$. In this case, the mapping $g\mapsto \Lambda_g \Psi \Lambda_{g^{-1}} m$ is weak*-continuous. In fact, if $f\in C (G)$, then
\begin{eqnarray*}
&& \left| \int f(x) \de (\Lambda_{g} \Psi \Lambda_{g^{-1}} m) (x) - \int f(x) \de (\Lambda_{h} \Psi \Lambda_{h^{-1}} m) (x) \right| \\
&& \qquad \qquad \leq \left| \int f(gx) \de (\Psi \Lambda_{g^{-1}} m - \Psi \Lambda_{h^{-1}} m) (x) \right| + \left| \int [f(gx) - f(hx)] \de (\Psi \Lambda_{h^{-1}} m) (x) \right| \\
&& \qquad \qquad \leq \no{f}_\infty \no{\Psi \Lambda_{g^{-1}} m - \Psi \Lambda_{h^{-1}} m}_M + \no{f(g\cdot) - f(h\cdot)}_\infty \no{\Psi \Lambda_{h^{-1}} m}_M \\
&& \qquad \qquad \leq \no{f}_\infty \no{\Psi} \no{\Lambda_{g^{-1}} m - \Lambda_{h^{-1}} m}_M + \no{f(g\cdot) - f(h\cdot)}_\infty \no{\Psi} \no{m}_M ,
\end{eqnarray*}
and the mappings $g\mapsto \Lambda_{g^{-1}} m$ and $g\mapsto f(g\cdot)$ are both norm continuous.

If $m \in L^1 (G)$, then for all $f\in C (G)$ the mapping $g\mapsto \int f(x) \de (\Lambda_{g} \Psi \Lambda_{g^{-1}} m) (x)$ is $\mu$-integrable, and we have
$$
\left| \int \left[ \int f(x) \de (\Lambda_{g} \Psi \Lambda_{g^{-1}} m) (x) \right] \de g \right| \leq \no{f}_\infty \no{\Psi} \no{m}_M .
$$
This shows that there exists a measure $\breve{\Psi} m \in M(G)$ such that
\begin{equation*}
\int \left[ \int f(x) \de (\Lambda_{g} \Psi \Lambda_{g^{-1}} m) (x) \right] \de g = \int f(x) \de (\breve{\Psi} m ) (x) \qquad \forall f\in C (G),
\end{equation*}
and the linear mapping $m \mapsto \breve{\Psi} m$ is bounded from $L^1 (G)$ into $M(G)$.

Suppose $m \in L^1 (G) \cap M_1^+ (G)$. If $\{ f_n \}_{n\in \N} \subset C (G)$ is such that $f_n \geq 0$ and $f_n \uparrow 1$, then $\int f_n(x) \de (\Lambda_{g} \Psi \Lambda_{g^{-1}} m) (x) \uparrow 1$ since $\Lambda_{g} \Psi \Lambda_{g^{-1}} m \in M_1^+ (G)$. Therefore, $\int f_n (x) \de (\breve{\Psi} m ) (x) \uparrow 1$ by dominated convergence theorem, thus showing that $\breve{\Psi} m \in M^+_1 (G)$.

Clearly, $\Lambda_g \breve{\Psi} m = \breve{\Psi} \Lambda_g m$ for all $m\in L^1 (G)$. In particular, the map $g\mapsto \Lambda_g \breve{\Psi} m$ is continuous, hence $\breve{\Psi} m \in L^1 (G)$ by Theorem 1.6 in \cite{Rudin59} (which is unaltered even if $G$ is not Abelian). By a result of Wendel \cite{Wendel52}, there exists $\nu\in M^+_1 (G)$ such that $\breve{\Psi}(\phi) = \nu \ast \phi$ for all $\phi \in L^1 (G)$.

Since $\Theta_\E (\sh) , \Theta_\F (\sh) \subseteq L^1 (G)$ (see e.g. \cite{HeLaPePuYl03}), and $\Lambda_g \Theta_\F = \Psi \Lambda_g \Theta_\E$ by covariance of $\E$ and $\F$, we have $\Theta_\F = \breve{\Psi} \circ\Theta_\E$. We conclude that $\F=\Enu$, and thus, (i) implies (ii).
\end{proof}

Let the probability measure $\nu$ in \eqref{eq:Enu} be the Dirac measure $\delta_{x}$ in some point $x\in\group$. In this case, the corresponding observable $\E_{\delta_{x}}$ has the form
\begin{equation}
\E_{\delta_x}(X)=\E(Xx^{-1})\equiv\E_x(X) \, .
\end{equation}
This observable, which we denote by $\E_x$, is therefore just a translated version of $\E$. It is clear that $\E \csim \E_x$. If $G$ is an Abelian group, then $\E_{x}(X)=\E(x^{-1}X)=U(x^{-1})\E(X)U(x)$ and therefore $\E_{x}\sim_U \E$.

\begin{proposition}\label{prop:post-approx-sharp}
Let $\E\in\obs$ and $\nu\in\probG$. The following conditions are equivalent:
\begin{itemize}
\item[(i)] $\Enu$ is approximately sharp.
\item[(ii)] $\E$ is approximately sharp and $\nu=\delta_{x}$ for some $x\in\group$ (i.e. $\Enu=\E_x$). 
\end{itemize}
\end{proposition}

\begin{proof}
Let us first note that since $\group$ is Hausdorff and second countable, it is metrizable and we can choose a left invariant metric $d$ for $\group$, i.e., $d(gx,gy)=d(x,y)$ for every $g,x,y\in\group$ (see, for instance, Theorem 8.3 in \cite{AHAI63} for this fact). 

Suppose that (i) holds. We make a counter assumption that $\nu$ is not a Dirac measure, i.e., there are two different points $x,y$ in the support of $\nu$. Denote $r:=\frac{1}{5} d(x,y)$. Then the open balls $B(x;r)$ and $B(y;r)$ are disjoint and have
positive $\nu$-measure. Thus, if we set $\alpha=\min\{\nu(B(x;r)),\nu(B(y;r))\}$, then $\alpha>0$.
As the metric $d$ is left invariant, we have $gB(x;r)=B(g x;r)$ for any
$g\in\group$. Moreover, for any $g\in\group$, we have
$B(g x;r)\cap B(x;r)=\emptyset$ or
$B(g x;r)\cap B(y;r)=\emptyset$. This
implies that $\nu(g B(x;r))\leq 1-\alpha$ for all
$g\in\group$. It then follows that
$$
\no{\Enu(B(x;r))}=\no{\int_{\group}\nu(g^{-1}B(x;r))\ d \E(g)}\leq
\sup_{g\in\group}\nu(g^{-1}B(x;r)) \cdot \no{\E(\group)} \leq
1-\alpha<1.
$$
This is in contradiction with our assumption that (i) holds. Hence, $\nu$ is the Dirac measure $\delta_{x}$ for some point $x\in\group$. For every $x\in G$ and $X\in\borG$, we have 
\begin{equation}\label{eq:norm-gX}
\no{\E_{\delta_x}(Xx)}=\no{\E(X)} \, .
\end{equation}
Therefore, $\no{\E(X)}=1$ for every open set $X\subseteq\group$ since the right multiplication is a homeomorphism. This shows that (i)$\Rightarrow$(ii).

The fact that (ii)$\Rightarrow$(i) is clear from \eqref{eq:norm-gX}.
\end{proof}

As a consequence of Propositions \ref{prop:post} and \ref{prop:post-approx-sharp}, we get the following result.

\begin{proposition}\label{prop:3}
Let $\E\in\obs$ be an approximately sharp observable. Then $\E$ is postprocessing clean in $\obs$. Another observable $\F\in\obs$ is postprocessing equivalent with $\E$ if and only if $\F=\E_x$ for some $x\in\group$.
\end{proposition}

\begin{proof}
Let $\E\in\obs$ be an approximately sharp observable. If $\F\in\obs$ is such that $\E\coar\F$, then by Prop.~\ref{prop:post} we have $\E=\F_{\nu}$ for some $\nu\in\probG$.  But Prop.~\ref{prop:post-approx-sharp} now implies that $\nu=\delta_x$ for some $x\in\group$, which means, in particular, that $\E\csim\F$. Therefore, $\E$ is postprocessing clean. This reasoning also proves the second claim.
\end{proof}

\subsection{Preprocessing}\label{sec:pre}

Let $\Phi:\trh\to\trh$ be a linear mapping such that its adjoint $\Phi^*:\lh\to\lh$ is completely positive and $\Phi^*(I)=I$. We say that $\Phi$ is a {\it channel} and denote the set of all channels by $\ch$. Let $\CPH$ be the set of the adjoints of channels, that is, $\CPH$ consists of normal completely positive maps $\Xi : \lh \frecc \lh$ such that $\Xi(I) = I$. 

\begin{definition}\label{def:pre}
Let $\E,\F \in\obs$. We say that $\F$ is a \emph{preprocessing} of
$\E$, and denote $\F\dis\E$, if there exists a channel
$\Phi:\trh \to \trh$ such that 
\begin{equation}\label{eq:pre}
\Theta_\F=\Theta_\E\circ\Phi.
\end{equation}
\end{definition}

Written in terms of the adjoint channel $\Phi^\ast$, Eq.~\eqref{eq:pre} amounts to say that $\F (X) = \Phi^\ast (\E(X))$ for all $X\in\borG$.

The relation $\dis$ on $\obs$ is clearly a preorder and we denote by $\dsim$ the induced equivalence relation. 

\begin{definition}\label{def:pre-clean}
An observable $\E\in\obs$ is \emph{preprocessing clean in $\obs$} if for every $\F\in\obs$, the following implication holds: 
\begin{equation}\label{eq:pre-clean}
\E\dis\F\ \Rightarrow\ \E\dsim\F \, .
\end{equation}
\end{definition}

We have some simple connections of the preprocessing relation to the other relations.

\begin{proposition}\label{p6}
If $\E,\F\in\obs$ and $\E\sim\F$, then $\E\dsim \F$. 
\end{proposition}

\begin{proof}
The relation $\E\sim\F$ means that $\F(\cdot)=W\E(\cdot)W^*$ for some unitary operator $W$. Define unitary channels $\Phi_1(T)=W^*TW$ and $\Phi_2(T)=WTW^\ast$. Then $\Theta_\F=\Theta_\E\circ\Phi_1$ and $\Theta_\E=\Theta_\F\circ\Phi_2$. Hence, $\E\dsim \F$. 
\end{proof}

\begin{proposition}\label{prop:pre-approx-sharp}
Let $\E,\F\in\obs$ and $\F\dis\E$. If $\F$ is approximately sharp, then also $\E$ is approximately sharp.
\end{proposition}

\begin{proof}
For any $X\in\borO$, we get
\begin{equation*}
\no{\F(X)}=\sup_{\varrho\in\sh}\tr{\varrho\F(X)}=\sup_{\varrho\in\sh}\tr{\Phi(\varrho)\E(X)}\leq \sup_{\varrho\in\sh}\tr{\varrho\E(X)}=\no{\E(X)}.
\end{equation*}
Thus, $\no{\F(X)}=1$ implies that $\no{\E(X)}=1$.
\end{proof}

A channel $\Phi$ is called $U$-\emph{covariant} if
\begin{equation}
\Phi(U(g) \varrho U(g)^\ast) = U(g) \Phi(\varrho) U(g)^\ast 
\end{equation}
for every $\varrho\in\sh$ and $g\in G$.

\begin{proposition}\label{prop:channel-covariant}
Let $\E,\F\in\obs$ and $\F\dis\E$. Then there is a $U$-covariant channel $\Phi$ such that \eqref{eq:pre} holds. 
\end{proposition}

\begin{proof}
Let $\Phi$ be a channel such that \eqref{eq:pre} holds. Since $\E$ and $\F$ are covariant, we have
\begin{equation}\label{10}
\Theta_\F(U(g)^\ast \varrho U(g))=\Theta_\E(U(g)^\ast \Phi(\varrho) U(g))
\end{equation}
for every $\varrho\in\sh$ and $g\in G$. We define a linear mapping $\breve{\Phi}:\trh\to\trh$ by setting
\begin{equation}
\breve{\Phi}(\varrho)=\int U(g)\Phi(U(g)^\ast \varrho U(g)) U(g)^\ast\ dg \, .
\end{equation}

Let $\Phi^*$ and $\breve\Phi^*$ be the adjoints of $\Phi$ and $\breve\Phi$.
Hence,
\begin{equation}
\breve{\Phi}^*(A)=\int U(g)\Phi^*(U(g)^\ast A U(g)) U(g)^\ast\ dg \, ,\quad A\in\lh \, .
\end{equation}
The fact that $\Phi^*$ is completely positive equals with 
\begin{equation}\label{CP}
\sum_{i,j=1}^n\langle\psi_i|\Phi^*(A_i^*A_j)\psi_j\rangle\ge 0
\end{equation}
for all $n=1,2,...,$ $A_1,...,A_n\in\lh$, and $\psi_1,...,\psi_n\in\hh$.
Since
$$
\sum_{i,j=1}^n\langle\psi_i|\breve\Phi^*(A_i^*A_j)\psi_j\rangle
=\int \left[
\sum_{i,j=1}^n\langle U(g)^*\psi_i|\Phi^*
\big((A_i U(g))^*(A_j U(g))
\big) U(g)^*\psi_j\rangle\right]\ dg
\ge 0
$$
it follows that $\breve\Phi^*$ is completely positive.
But  $\breve\Phi^*(I)=I$
so that
$\Breve\Phi$ is a channel.
Using the invariance of the Haar integral, it is straightforward to verify that $\breve{\Phi}$ is $U$-covariant. By Eq.~\eqref{10}, $\Theta_{\F} = \Theta_{\E} \circ \breve\Phi$.
\end{proof}


Proposition \ref{prop:channel-covariant} shows that the investigation of the preprocessing relation reduces to the study of covariant channels. We will need the following general result in order to characterise the set of $U$-covariant elements in $\CPH$. For a proof, see e.g. \cite{CaHeTo08}. 

\begin{proposition}\label{prop:PhiEst}
Suppose $U$ is a unitary representation of $G$ in a Hilbert space $\hh$, and $\Xi\in\CPH$ is a $U$-covariant channel. There exists a separable Hilbert space $\kk$, a unitary representation $D$ of $G$ in $\kk$, and an isometry $W : \hh \frecc \kk\otimes\hh$ such that 
\begin{eqnarray}
WU(g) &=& (D(g) \otimes U(g)) W  \qquad \forall g\in\group \, , \nonumber \\
\Xi (A) &=& W^\ast (I\otimes A) W \qquad \forall A\in\lh \, . \label{PhiEsteso}
\end{eqnarray}
\end{proposition}

We apply Proposition \ref{prop:PhiEst} in Section \ref{sec:phase-pre} in the study of phase observables.

\section{Optimal phase observables}\label{sec:phase}

\subsection{Structure of phase observables}\label{sec:phase-structure}

Let  $\hi$ be a complex Hilbert space with an orthonormal basis $\{\ket{n}\,|\,n\in\N\}$, $\N=\{0,1,2,...\}$. We define the lowering, raising, and number operators as
$$
a:=\sum_{n=0}^\infty\sqrt{n+1}\, \kb{n}{n+1},\qquad
a^*:=\sum_{n=0}^\infty\sqrt{n+1}\, \kb{n+1}{n},\qquad 
N:=a^*a=\sum_{n=0}^\infty n\, \kb{n}{n},
$$
respectively. Physically the Hilbert space $\hi$ and the above operators are associated with
a single-mode optical field. The vectors $\ket{n}$ are called \emph{number states}.

Coherent states $\ket{z}:=e^{-|z|^2/2}\sum_{n=0}^\infty \frac{z^n}{\sqrt{n!}}\ket{n}$, where $z\in\C$,
describe laser light which is quasimonochromatic and thus can be approximated as a single-mode system;
here $|z|$ is the energy and $\arg z$ is the phase parameter of $\ket{z}$.

An observable $\E$ describing a phase parameter measurement should have the interval $[0,2\pi)$ as its outcome space. For convenience, we use one dimensional torus $\T=\{t\in\C:\abs{t}=1\}$ as an equivalent description. Let $\ket{z}$ be a coherent state and $t\in\T$. Then $\ket{tz}$ is another coherent state, now having phase parameter $\arg z+\arg t$ (addition modulo $2\pi$). Therefore, we require that $\E$ describing a phase parameter measurement satisfies the condition
\begin{equation}\label{eq:phase-condition-z}
p^{\E}_{\ket{tz}}(X)=p^{\E}_{\ket{z}}(t^{-1} X)
\end{equation}
for all $z\in\C$, $t\in\T$ and $X\in\borT$. 

As shown in \cite{LaPe02}, an observable $\E:\borT\to\mathcal{L(H)}$ satisfies condition \eqref{eq:phase-condition-z} if and only if
\begin{equation}\label{eq:phase-condition}
U(t) \E(X)U(t)^\ast=\E(tX)
\end{equation}
for all $t\in\T$ and $X\in\borT$, where $U$ is the {\em number representation} of $\T$ in $\hh$, i.e.,
$$
U(t) \ket{n} = t^n \ket{n} \quad \forall n\in\N .
$$
Clearly, $U(t) \ket{z} = \ket{tz}$
We take this covariance condition as the definition for phase observables. 

\begin{definition}
An observable $\E:\borT\to\mathcal{L(H)}$ is a \emph{phase observable} if it satisfies the covariance condition \eqref{eq:phase-condition}.
\end{definition}

The following phase theorem characterizing phase observables has been proved in various different methods in \cite{CaDeLaPe02,Holevo83,LaPe99}.

\begin{theorem}[Phase Theorem]\label{PhTh}
An observable $\E:\borT\to\lh$ is a phase observable if and only if 
\begin{equation*}
\langle{m}|\E(X)|{n}\rangle = c_{m,n}\,\int_X t^{m-n} dt
=c_{m,n}\,\int_{\arg X} e^{i(m-n)\theta}\frac{d\theta}{2\pi}
 \qquad \forall m,n\in\N \, ,
\end{equation*}
where the {\em phase matrix} $(c_{m,n})$ is a positive semidefinite complex $\N\times\N$-matrix and $c_{n,n}=1$ for all $n\in\N$.

If $\E_1$, $\E_2$ are phase observables with phase matrices $(c^1_{m,n})$ and $(c^2_{m,n})$, then $\E_1 \sim_U \E_2$ if and only if there exists a sequence $\{\lambda_n \}_{n\in\N} \in \T$ such that $c^1_{m,n} = \lambda_n \overline{\lambda_m} c^2_{m,n}$.
\end{theorem}

Using Theorem \ref{PhTh} we can easily see some general properties of phase observables. For instance, the probability measure $p^{\E}_{\ket{n}}$ of a phase observable $\E$ in a number state $\ket{n}$ is uniformly distributed, which is expected as number states do not have specific phase.  
Also, a phase observable is never sharp (a projection valued measure); see e.g.\ \cite{LaPe99}.

We notice that $U$-equivalence for two phase observables is, in general, stronger requirement than equivalence. This is demonstrated in the following proposition.

\begin{proposition}\label{prop:not-U-equivalent}
There exist phase observables $\E_1$ and $\E_2$ such that $\E_1$ is equivalent to $\E_2$, but $\E_1$ is not $U$-equivalent to $\E_2$.
\end{proposition}
 
\begin{proof}
Fix $z\in\C$, $0<|z|<1$, and define phase observables $\E_1$ and $\E_2$ by setting
\begin{eqnarray*}
\E_1(X) &=& \int_X dt\,I+z\int_Xt^{-1}dt\, \kb01 + \overline z\int_Xt\, dt\, \kb10 \, , \\
\E_2(X) &=& \int_X dt\,I+z\int_Xt^{-1}dt\, \kb23 + \overline z\int_Xt\, dt\, \kb32
\end{eqnarray*}
for all $X\in\borT$. Define a unitary operator $W$ as
\begin{equation*}
W=\kb 20 + \kb 02 + \kb 13 + \kb 31 + \sum_{n=4}^\infty\kb n n \, .
\end{equation*}
Then $\E_2=W\E_1W^*$. From Theorem \ref{PhTh} one sees easily that $\E_1$ and $\E_2$ cannot be $U$-equivalent.
\end{proof}

\subsection{Canonical phase observable}\label{sec:canonical}

The \emph{canonical phase observable} $\Ecan$ is the phase observable determined by the phase matrix $c_{m,n}=1$ for all $m,\,n\in\N$. If $\E$ is another phase observable with a phase matrix $(c_{m,n})$, then $\E\sim_U \Ecan$ if and only if $|c_{m,n}| = 1$ for all $m,\,n\in\N$; see \cite{LaPe00}.

The canonical phase observable $\Ecan$ has some properties which make it special among all phase observables. For example, it is the only phase observable (up to equivalence)  which generates
the number shifts \cite{LaPe00}. A drawback of $\Ecan$ is that a realistic measurement scheme for it is not known.

We will see that $\Ecan$ is optimal phase observable in all the four different ways we listed in Section \ref{sec:intro}. However, (perhaps surprisingly) it is not the unique phase observable having this feature. 

\subsection{State generated phase observables}\label{sec:state-generated}

Let  $D(z):=e^{za^*-\overline{z}a}$,  $z\in\mathbb C$, be the shift operator, and denote by $\shd$ the convex set of all {\it diagonal states}, that is, the states of the form
$T=\sum_{n=0}^\infty\lambda_n|n\rangle\langle n|$, where $\lambda_n\ge 0$ for all $n\in\N$ and $\sum_{n=0}^\infty\lambda_n=1$. Obviously, the extremal elements of $\shd$ are the one-dimensional projections of the form $\kb{n}{n}$, $n\in\N$.

The \emph{phase observable generated by $T\in\shd$}
(or a \emph{phase space phase observable}), denoted by $\E_T$, is defined by
\begin{equation}\label{eq:ET}
\E_T(X)=\frac{1}{\pi}\int_{\arg X}\int_0^\infty D\big(re^{i\theta}\big)T D\big(re^{i\theta}\big)^*r\,dr\,d\theta
\hspace{1cm}\text{(weakly)}
\end{equation}
for all $X\in{\mathcal B}(\mathbb T)$ \cite{LaPe99}.
Using the decompostion of $T$, one gets
$$
\E_T(X)=\sum_{n=0}^\infty\lambda_n \E_{|n\rangle}(X) 
\hspace{1cm}\text{(weakly) for all $X\in{\mathcal B}(\mathbb T)$,}
$$
where we have denoted $\E_{\kb{n}{n}}$ briefly by $\E_{\ket{n}}$. Note that the operator $T$ in formula \eqref{eq:ET} must be diagonal in the number basis so that the corresponding observable would be phase shift covariant \cite{LaPe99}.

State generated phase observables are important since some of them have been measured. Indeed, $\E_{\ket{0}}$ can be seen as an angle margin observable of a phase space observable associated to a $Q$-function, and this has been measured by Walker and Carrol for a coherent state input field \cite{WaCa86}. Moreover, in principle all state generated phase observables can be measured using an eight-port homodyne detector \cite{KiLa08}.

We will see that state generated phase observables are approximately sharp and postprocessing clean. They are not extremal nor preprocessing clean. 

\subsection{Approximately sharp phase observables}\label{sec:phase-sharp}

Let $C(\mathbb T)$ be a Banach space of continuous complex functions on $\mathbb T$
equipped with the sup norm. By the Riesz representation theorem, the topological dual $M(\mathbb T)$ of $C(\mathbb T)$ consists of regular complex Borel measures on $\mathbb T$ (or Baire measures).
We equip $M(\mathbb T)$ with the weak-star topology. In this topology, a net $\{p_i\}_{i\in\mathcal I}\subset M(\mathbb T)$ converges to a point $p\in M(\mathbb T)$
if $\lim_{i\in\mathcal I}\int f dp_i=\int f dp$ for all $f\in C(\mathbb T)$;
we denote $p=\text{w*-lim}_{i\in\mathcal I}p_i$.

Let $\delta_t:\,C(\mathbb T)\to\C,\; f\mapsto \delta_t(f)=f(t)$ be the Dirac distribution (or Dirac measure) concentrated on $t\in\mathbb T$. Let us write the parameter $z$ of a coherent state $\ket z$ in the form $z=rt$, where
$r=|z|\in\R$ and $t=z/|z|\in\T$. Let us then recall the following result, proved in \cite{LaPe00}.

\begin{proposition}\label{prop:9}
Let $\E$ be a phase observable with a phase matrix $(c_{m,n})$, and let $u\in\T$. Then
$$
\mathop{\text{\rm w*-lim}}_{r\to\infty}p_{|rt\rangle}^\E=\delta_{tu}
$$
if and only if
$$
\lim_{m\to\infty}c_{m,m+k}=u^k
$$
for all $k=1,2,3,...$. 
\end{proposition}

This leads to the following useful conclusion.

\begin{proposition}\label{prop:phase-u}
Let $\E$ be a phase observable with a phase matrix $(c_{m,n})$. If there exists $u\in\T$ such that
\begin{equation}\label{eq:lim}
\lim_{m\to\infty}c_{m,m+k}=u^k
\end{equation}
for all $k=1,2,3,...$, then $\E$ is approximately sharp.
\end{proposition}

\begin{proof}
If condition \eqref{eq:lim} holds, then 
$\mathop{\text{\rm w*-lim}}_{r\to\infty}p_{|rt\rangle}^\E=\delta_{tu}$ by the previous proposition.
Let $X\in \borT$ be a nonempty open set and
choose $t$ such that $tu\in X$.
Let $f\in C(\mathbb T)$ be such that $f(tu)=1$, the support of $f$ is contained in $X$ and $0\le f\le 1$.
For any $\epsilon>0$ there exists $r_\epsilon\ge 0$
such that $\int f dp_{|rt\rangle}^\E>1-\epsilon$ for all $r\ge r_\epsilon$.
Since, for all $r\ge r_\epsilon$, 
$$
1=p_{|rt\rangle}^\E(X)+p_{|rt\rangle}^\E(\mathbb T\setminus X)\ge
\int f dp_{|rt\rangle}^\E
+p_{|rt\rangle}^\E(\mathbb T\setminus X)
>1-\epsilon+p_{|rt\rangle}^\E(\mathbb T\setminus X)
$$
it follows that $p_{|rt\rangle}^\E(\mathbb T\setminus X)<\epsilon$ for all $r\ge r_\epsilon$.
Thus, $\lim_{r\to\infty}p_{|rt\rangle}^\E(\mathbb T\setminus X)=0$ and, hence,
$\lim_{r\to\infty}p_{|rt\rangle}^\E(X)=1$. This implies that $\|\E(X)\|=1$.
\end{proof}

Note that (in the context of the above proof) 
$\lim_{r\to\infty}p_{|rt\rangle}^\E(X)=\delta_{tu}(X)$ {\it does not} hold for all open $X$.
Take, e.g., $X=\mathbb T\setminus\{tu\}$ and note that then $p_{|rt\rangle}^\E(X)=1$ since
$p_{|rt\rangle}^\E$ is absolutely continuous measure with respect to the Haar measure.

Since the canonical phase observable $\Ecan$ corresponds to the phase matrix $c_{m,n}=1$, it follows from Proposition \ref{prop:phase-u} that $\Ecan$ is approximately sharp (see also \cite{HeLaPePuYl03} for a different proof). 
\begin{proposition}\label{prop:11}
All state generated phase  observables $\E_T$, $T\in\shd$, are approximately sharp. 
\end{proposition}

\begin{proof}
Let $\E_T$ be the phase observable generated by $T=\sum_{s\in\N}\lambda_s \kb{s}{s}\in\shd$.
Since
$
\E_T=\sum_{s=0}^\infty\lambda_s \Es
$ (weakly),
the phase matrix elements of $\E_T$ are $c^T_{m,n}=\sum_{s=0}^\infty\lambda_s c_{m,n}^{\ket s}$.
It has been shown in \cite{LaPe00} that $\lim_{m\to\infty}c^{\ket s}_{m,m+k}=1$ for all $k\in\N$ and $s\in\N$. Since $| c^{\ket s}_{m,m+k} | \leq 1$, we have $\lim_{m\to\infty} c^T_{m,m+k} = 1$ for all $k\in\N$.
The claim follows then from Prop. \ref{prop:phase-u} with $u=1$.
\end{proof}

Finally, let us notice that, as seen in the proof of Proposition \ref{prop:phase-u},
$$
\lim_{r\to\infty}\langle rt|\E_{T}(X)|rt\rangle=1
$$
where $X\subseteq\mathbb T$ is open set and $t\in X$.
This means that in the classical limit when the energy of the state $|z\rangle$ is large, the phase parameter of $|z\rangle$ can be determined accurately by measuring $\E_T$.
However, there is no phase observable $\E$ for which $\langle rt|\E(X)|rt\rangle\approx1$ for finite energies $r=|z|$. There exist (phase shift covariant) generalized operator measures satisfying this condition; these measures describe measurements where only a restricted class of state preparations are available \cite{Pellonpaa01}.

\subsection{Extremal phase observables}\label{sec:phase-convex}

Let $\CC$ be the convex set of all phase matrices and let $\ext\CC$ be the set of the extremal elements. The correspondence between the convex sets $\CC$ and $\obs$ preserves the convex structures. Especially, a phase observable $\E$ determined by a phase matrix $(c_{n,m})$ is extremal in $\obs$ if and only if $(c_{n,m})$ is extremal in $\CC$.

Let $(c_{n,m})$ be a phase matrix. It is always possible to choose a Hilbert space $\hik$ and a sequence of unit vectors $\{ \eta_n \}_{n\in\N}$ which is total in $\hik$ such that $c_{n,m}=\ip{\eta_n}{\eta_m}$ for all $n,m\in\N$; see \cite{CaDeLaPe02}. The dimension of $\hik$ is uniquely determined by $(c_{n,m})$ and we call $\dim\hik$ the \emph{rank} of  $(c_{n,m})$. Let us denote by $\mathcal T(\hik)$ the Banach space of trace-class operators on $\hik$. As proved in \cite{CaHePeTo08}, \cite{KiPe08}, the phase matrix $(c_{n,m})$ is extremal if and only if the trace-class closure of ${\rm lin}_{\C}\left\{\kb{\eta_n}{\eta_n} : n\in\N \right\}\subseteq\mathcal T(\hik)$ is $\mathcal T(\hik)$. 

\begin{example}
Let $\xi\in\C$, $\abs{\xi}\le 1$. Define unit vectors $\eta^\xi_{2n}:=\ket 0$ and $\eta^\xi_{2n+1}:=\xi\ket0+\sqrt{1-|\xi|^2}\ket 1$ for all $n\in\N$ and a phase matrix $(c_{n,m}^\xi)$ by $c_{n,m}^\xi:=\langle\eta^\xi_n|\eta^\xi_m\rangle$. Thus, $(c_{n,m}^\xi)$ is the phase matrix of the so-called {\it chess-board phase observable} $\E_\xi$; see \cite{LaPe00}. Assume that $\abs{\xi}\neq 1$. Then the Hilbert space
 $\hik$ is $\C\ket 0+\C\ket 1$ and the rank of $(c_{n,m}^\xi)$ is $2$.
 Moreover, 
 $$
 \kb01\notin
 {\rm lin}_{\C}\left\{\kb{\eta^\xi_n}{\eta^\xi_n} : n\in\N \right\}=\C\kb00+\C\big(
 \xi\kb01+\overline\xi\kb10+\sqrt{1-|\xi|^2}\kb11\big)
 $$
 so that $\E_\xi$ is not extremal. If $\abs{\xi}=1$, then the rank of $(c_{n,m}^\xi)$ is $1$ and $\E_\xi$ is automatically extremal; see Proposition 5 in \cite{CaHePeTo08}. In addition, $|c_{n,m}^\xi|=1$ for all $n,\,m\in\N$ so that $\E_\xi\sim\Ecan$ \cite{LaPe00}. For more examples of extremal phase observables, we refer to \cite{KiPe08}.
 \end{example}

\begin{proposition}\label{prop:creale}
If the phase matrix $(c_{m,n})$ has rank greater than $1$ and $c_{m,n}\in\R$ for all $m,n\in\N$, then 
$(c_{m,n})\notin\ext\CC$, that is, the corresponding phase observable $\E$ is not extremal.
\end{proposition}

\begin{proof}
Let $\eta_m$, $\eta_n$ be two linearly independent vectors in $\kk$, and define a (nonzero) bounded operator $B = \kb{\eta_m}{\eta_n} - \kb{\eta_n}{\eta_m}$ on $\kk$. For every $j\in\N$, we have
$$
\tr{B \kb{\eta_j}{\eta_j}} = c_{n,j} c_{j,m} - c_{m,j} c_{j,n} = c_{n,j} c_{j,m} - \overline{c_{n,j}}\overline{c_{j,m}} = c_{n,j} c_{j,m} - c_{n,j} c_{j,m}= 0 \, .
$$
By the Hahn-Banach theorem this implies that the set  ${\rm lin}_{\C}\left\{\kb{\eta_n}{\eta_n} : n\in\N \right\}$ is not dense in $\mathcal T(\hik)$. Therefore, using the above cited criterion the phase matrix $(c_{m,n})$ is not extremal.
\end{proof}

Fix $s\in\N$. Elements of the phase matrix $\big(c_{m,n}^{\ket{s}}\big)$ of the state generated phase observable $\Es$ are of the form
$$
c_{m,n}^{\ket{s}}=\int_0^\infty f^s_m(x) f^s_n(x)e^{-x}dx \, ,
$$
where
$$
f^s_n(x)=(-1)^{\max\{0,s-n\}}\sqrt{\frac{(\min\{n,s\})!}{(\max\{n,s\})!}}x^{|s-n|/2}L^{|s-n|}_{\min\{n,s\}}(x)
$$
and 
$$
L^\alpha_k(x)=\sum_{l=0}^k(-1)^l{{k+\alpha}\choose{k-l}}\frac{x^l}{l!}
$$
is the associated Laguerre polynomial; see \cite{LaPe99, LaPe00}.
Hence, the $\eta_n$-vectors of $\big(c_{m,n}^{\ket{s}}\big)$ are the unit vectors $f^s_n$, $n\in\N$, of $L^2(\R_+,\,e^{-x}dx)$, where $\R_+$ is the set of positive reals. All the functions $f^s_n$, $n\in\N$, are linearly independent, and therefore the rank of $(c_{m,n}^{\ket s})$ is infinite. It is clear that $c_{m,n}^{\ket{s}} \in\R$ for all $m,n\in\N$, and hence the phase observable $\Es$ is not extremal by Proposition \ref{prop:creale}.

Let $\E_T$ be the phase observable generated by a state $T=\sum_{s=0}^{\infty}\lambda_s \kb{s}{s}\in\shd$. Then $\E_T=\sum_{s=0}^\infty\lambda_s \Es$ and hence, $\E_T$ is not extremal. We thus have the following conclusion.

\begin{proposition}
Any $\E_T$, $T\in\shd$, is not extremal in the convex set of all phase observables.
\end{proposition}

For completeness, we note that the phase matrix elements of $\E_T$ are $c^T_{m,n}=\sum_{s=0}^\infty\lambda_s c_{m,n}^{\ket s}=\ip{f^T_m}{f^T_n}$, where $f^T_n=\sum_{s=0}^\infty \sqrt{\lambda_s}\,f_n^s\otimes\ket{s}$ is a unit vector
of the Hilbert space $L^2(\R_+,\,e^{-x}dx)\otimes\mathcal H$, and it follows that $\rank(c_{m,n}^T)=\infty$.

\begin{remark}
Since the physically relevant phase observables $\E_T$ and $\Ecan$ all have
real-valued phase matrices, it would be interesting to study the extremals of the smaller convex set of real-valued phase matrices $\CC_\R=\{(c_{m,n})\in\CC\,|\,c_{m,n}\in\R,\,m,n\in\N\}$.
The complete characterization of such extemals is given in \cite[Theorem 1]{KiPe08}.
The canonical phase observable is extremal in $\CC_\R$, but the question of the extremality of the state generated phase observables is open. The method of the proof of Proposition \ref{prop:creale} cannot be directly applied in this case as the operator $\kb{\eta_m}{\eta_n} - \kb{\eta_n}{\eta_m}$ is trivial on the real Banach space $\mathcal{T}_s (\hik)$ of selfadjoint trace class operators on $\kk = \overline{{\rm lin}_{\R}\left\{\eta_n : n\in\N \right\}}$.
\end{remark}

\subsection{Postprocessing}\label{sec:phase-post}

As explained in Section \ref{sec:phase-sharp}, the observables $\Ecan$ and $\E_T$ are approximately sharp. Therefore, it follows directly from Proposition \ref{prop:3} that they are all postprocessing clean in $\obs$. It still remains to check whether they are postprocessing equivalent.
Let us start with the following observation.

\begin{proposition}\label{prop:T=T'}
Let $T,\,T'\in\shd$. If $\E_T=\E_{T'}$, then $T=T'$.
\end{proposition}

\begin{proof}
Let $T=\sum_{s\in\N}\lambda_s\kb s s\in\shd$ and  $T'=\sum_{s\in\N}\lambda'_s\kb s s\in\shd$.
Thus, $c^T_{m,n}=\sum_{s=0}^\infty\lambda_s c_{m,n}^{\ket s}$ and
$c^{T'}_{m,n}=\sum_{s=0}^\infty\lambda'_s c_{m,n}^{\ket s}$.
Assume $\E_T=\E_{T'}$ so that $c^{T}_{m,n}=c^{T'}_{m,n}$ for all $m,\,n\in\N$.
For all $s,\,k\in\N$, 
\begin{eqnarray*}
c_{0,2k}^{\ket{s}}&=&(-1)^{s+\max\{0,s-{2k}\}}\sqrt{\frac{(\min\{{2k},s\})!}{s!(\max\{{2k},s\})!}}
\int_0^\infty x^{(s+|s-{2k}|)/2}L^{|s-{2k}|}_{\min\{{2k},s\}}(x)e^{-x}dx \\
&=&(-1)^{s+\max\{0,s-{2k}\}}
\frac{\big((|s-2k|+s)/2\big)!\,\Gamma(k)}
{\sqrt{(2k)!}\,s!\,\Gamma\big((|s-2k|-s)/2\big)}
\end{eqnarray*}
where $\Gamma$ is the Gamma function and we have used the following equation (see formula 7.414(11) in \cite{TISP80}):
$$
\int_0^\infty x^{\gamma-1}L^\alpha_n(x)e^{-x}dx=\frac{\Gamma(\gamma)\Gamma(1+\alpha-\gamma+n)}
{n!\,\Gamma(1+\alpha-\gamma)}\, ,\hspace{1cm}\gamma>0,\;\alpha,\,n\in\N.
$$
Since $\lim_{x\to-n}|\Gamma(x)|=\infty$ for all $n\in\N$,
it follows that
$$
c_{0,2k}^{\ket s}=0
$$ 
if and only if $0<k\le s$.
Thus,
$$
c^T_{0,2k}=\sum_{s=0}^{k-1}\lambda_s c_{0,2k}^{\ket s}=\sum_{s=0}^{k-1}\lambda'_s c_{0,2k}^{\ket s}.
$$
By induction $T=T'$. Indeed,
$c^T_{0,2}=\lambda_0 c_{0,2}^{\ket 0}=\lambda'_0 c_{0,2}^{\ket 0}$ implies that $\lambda_0=\lambda'_0$.
If $\lambda_s=\lambda'_s$ for all $s=0,1,...,k-1$ then
\begin{equation*}\
(\lambda_k-\lambda'_k)c_{0,2(k+1)}^{\ket k}=
c^T_{0,2(k+1)}-c^{T'}_{0,2(k+1)}=0.
\end{equation*}
Thus, $\lambda_k=\lambda'_k$.
\end{proof}

\begin{proposition}
Phase observables $\Ecan$ and $\E_T$ are postprocessing clean. They are all postprocessing non-equivalent.
\end{proposition}

\begin{proof}
We need to prove the second claim. For two approximately sharp phase observables $\E$ and $\E'$ with phase matrices $(c_{m,n})$ and $(c'_{m,n})$, Prop. \ref{prop:3}  implies (taking into account that $\T$ is Abelian group) that $\E\csim\E'$ equals $\E'(\cdot)=U(t)\E(\cdot)U(t)^*$ for some $t\in\T$. By Theorem \ref{PhTh} this means that $c'_{m,n}=t^{m-n}c_{m,n}$ for all $m,\,n\in\N$.

Let $T,\,T'\in\shd$ and assume that, for some $t\in\T$,
$\E_{T'}(\cdot)=U(t)\E_T(\cdot)U(t)^*$. We have seen in the proof of Prop. \ref{prop:11} that $\lim_{m\to\infty}c^T_{m,m+k} = \lim_{m\to\infty}c^{T'}_{m,m+k} = 1$ for all $k\in\N$. On the other hand, $c^{T'}_{m,m+k} = t^{-k} c^{T}_{m,m+k}$, from which $t = 1$ follows. Hence, $\E_T\csim\E_{T'}$ only if $\E_T=\E_{T'}$, i.~e.~only if $T = T'$ by Prop.~\ref{prop:T=T'}.

Assume then that $\Ecan\csim\E_T$. It follows that $|c^T_{0,2}|=1$. This is impossible since $c^T_{0,2}=\lambda_0 c_{0,2}^{\ket 0}$ (see the proof of Prop. \ref{prop:T=T'}) and
$$
c_{0,2}^{\ket{0}}=\int_0^\infty f^0_0(x) f^0_2(x)e^{-x}dx=
\frac1{\sqrt2}\int_0^\infty x e^{-x}dx=\frac1{\sqrt2}<1 \, .
$$
\end{proof}


\subsection{Preprocessing}\label{sec:phase-pre}

Let $\E$ be a phase observable determined by the phase matrix $(c_{m,n})\in\CC$.
As shown in \cite{Pellonpaa02}, we can define a $U$-covariant channel by formula
$$
\trh\ni T\mapsto \Phi_\E(T):=\sum_{m,n=0}^\infty c_{n,m}\langle m|T|n\rangle \kb m n \in\trh \, .
$$
Moreover, this channel satisfies 
$$
\Theta_\E = \Theta_{\Ecan} \circ \Phi_\E \, .
$$
This shows that every phase observable is a preprocessing of the canonical phase observable $\Ecan$. In particular, $\Ecan$ is preprocessing clean.

It remains to find the equivalence class of $\Ecan$ in the preprocessing relation. The preprocessing clean phase observables are exactly the phase observables belonging to this equivalence class.

\begin{proposition}\label{prop:PhiVCov}
If $\Xi\in\CPH$ is such that 
$$
\Xi(U(t) A U(t)^\ast) = U(t) \Xi(A) U(t)^\ast
$$
for all $t\in \T$ and $A\in\lh$, then there exists a sequence of vectors $\{ \phi_i \}_{i\in\N}$ in $\hh\otimes \hh$ with $\no{\phi_i} = 1$ such that
$$
\langle j | \Xi (A) | i \rangle = \int t^{i-j} \ip{\phi_j}{\lft I \otimes U(t) A U(t)^\ast \rgt \phi_i} \de t \quad \forall i,j \in\N \, .
$$
\end{proposition}

\begin{proof}
Let $\rho$ be the right regular representation of $\T$ in $\ldue{\T}$. Since $\T$ is compact and $\dim \hh = \infty$, every separable unitary representation of $\T$ is contained in the tensor product representation $\rho \otimes I$ acting in the space $\ldue{\T} \otimes \hh = \ldue{\T ; \hh}$. Therefore, $\kk = \ldue{\T ; \hh}$ and $D = \rho \otimes I$ is the most general choice in the context of Prop. \ref{prop:PhiEst}.

We define the following unitary operator $T : \ldue{\T} \otimes \hh \otimes \hh \frecc \ldue{\T} \otimes \hh \otimes \hh$
$$
[T(f \otimes u \otimes v)](t) = f(t) \otimes u \otimes U(t)^\ast v .
$$
Clearly,
$$
T (\rho \otimes I \otimes I) = (\rho \otimes I \otimes U) T = (D \otimes U) T .
$$
Therefore, every isometry $W : \hh \frecc \kk \otimes \hh$ intertwining $U$ with $D\otimes U$ is of the form $W = T\tilde{W}$, where $\tilde{W}$ is an isometry intertwining $U$ with $\rho \otimes I \otimes I$. The most general such $\tilde{W}$ is given by
$$
\tilde{W}\ket{n} = f_n \otimes \phi_n ,
$$
where $\{ \phi_i \}_{i\in\N}$ is a sequence of unit vectors in $\hh\otimes \hh$. By Prop. \ref{prop:PhiEst}, every $U$-covariant $\Xi\in\CPH$ is thus of the form
\begin{eqnarray*}
\langle j | \Xi (A) | i \rangle & = & \langle j |\tilde{W}^\ast T^\ast (I_\kk \otimes A) T \tilde{W}| i \rangle \\
& = & \ip{T(f_j \otimes \phi_j)}{(I_{\ldue{\T}} \otimes I_\hh \otimes A) T(f_i \otimes \phi_i)} \\
& = & \int t^{i-j} \ip{\phi_j}{\lft I \otimes U(t) A U(t)^\ast \rgt \phi_i} \de t
\end{eqnarray*}
for all $A\in\lh$ and $i,j\in\N$.
\end{proof}

\begin{proposition}\label{PropEcan}
If $\E$ is a phase observable with phase matrix $(c_{m,n})$, then the following are equivalent:
\begin{itemize}
\item[(i)] $\E \dsim \Ecan$
\item[(ii)] there exists $n_0\in\N$ and a sequence $\{ \lambda_n \}_{n\geq n_0}$ in $\T$ such that $c_{m,n} = \overline{\lambda}_m \lambda_n$ for all $m,n\geq n_0$. 
\end{itemize}
\end{proposition}

\begin{proof}
Let $\linf$ be the Banach space of bounded measurable functions on the unit circle $\T$ with the $\sup$ norm. We recall that, if $\E$ is a (not necessarily covariant) POVM based on $\T$ with values in $\hh$, then we can define a norm decreasing linear map $\E : \linf \frecc \lh$ in the following way. If $T\in\trh$, let $p^{\E}_T$ be the bounded complex measure on $\T$ given by
$$
p^{\E}_T (X) = \tr{T \E (X)} \quad \forall X\in\borT.
$$
If $f\in\linf$, the operator $\E (f)$ is then defined by
$$
\tr{T \E (f)} = \int f(t) \de p^{\E}_T (t) \quad \forall T\in\trh \, .
$$

For all $j\in\Z$, let $f_j (t) = t^j$. Then
\begin{eqnarray*}
\E (f_j) = \sum_{n\geq 0} c_{n,j+n} \kb{n}{j+n} && \textrm{ if } j\geq 0 \\
\E (f_j) = \sum_{n\geq 0} c_{n-j,n} \kb{n-j}{n} && \textrm{ if } j\leq 0 ,
\end{eqnarray*}
the sums converging in the weak sense. Clearly, $\E (f_j)^\ast = \E (f_{-j})$ for all $j\in\N$.

If $\Xi \in \CPH$, then it is easy to check that $\Xi (\E(f)) = \E^\Xi (f)$, where $\E^\Xi$ is the POVM given by
$$
\E^\Xi (X) = \Xi (\E(X)) \quad \forall X\in\borT \, .
$$

Suppose $\Xi$ is a $U$-covariant channel, and let $\{ \phi_i \}_{i\in\N}$ be a sequence of unit vectors in $\hh\otimes \hh$ associated to $\Xi$ as in Prop. \ref{prop:PhiVCov}. Then by Fourier analysis $\Ecan = \E^\Xi$ if and only if $\Ecan (f_j) = \E^\Xi (f_j)$ for all $j\in\N$.

For $j\in\N$, we have
\begin{eqnarray*}
\langle q | \E^\Xi (f_j) | p \rangle & = & \int t^{p-q} \ip{\phi_q}{\lft I \otimes U(t) \E (f_j) U(t)^\ast \rgt \phi_p} \de t \\
& = & \int t^{p-q} \sum_{n\geq 0} c_{n,j+n} \ip{\phi_q}{\lft I \otimes U(t) \kb{n}{j+n} U(t)^\ast \rgt \phi_p} \de t \\
& = & \sum_{n\geq 0} c_{n,j+n} \int t^{p-q-j} \ip{\phi_q}{\lft I \otimes \kb{n}{j+n} \rgt \phi_p} \de t \\
& = & \delta_{j,p-q} \sum_{n\geq 0} c_{n,j+n} \ip{\phi_q}{\lft I \otimes \kb{n}{j+n} \rgt \phi_p}
\end{eqnarray*}
for all $p,q\in \N$. Setting
\begin{equation}\label{def. phik}
\phi_k = \sum_{r\geq 0} \phi^r_k \otimes \ket{r}
\end{equation}
with $\phi^r_k \in \hh$ such that $\sum_{r\geq 0} \no{\phi^r_k}^2 = 1$, the last equation becomes
$$
\langle q | \E^\Xi (f_j) | p \rangle = \delta_{j,p-q} \sum_{n\geq 0} c_{n,j+n} \ip{\phi^n_q}{\phi^{j+n}_p} .
$$
This must be compared with
$$
\langle q | \Ecan (f_j) | p \rangle = \delta_{j,p-q} .
$$
The two expressions are the same if and only if
\begin{equation}\label{.}
\sum_{n\geq 0} c_{n,p-q+n} \ip{\phi^n_q}{\phi^{p-q+n}_p} = 1 \quad \textrm{for all } p\geq q\geq 0 .
\end{equation}
If $\{ \eta_n \}_{n\in\N}$ are unit vectors in $\hh$ such that $c_{n,m} = \ip{\eta_n}{\eta_m}$ (see \cite{CaDeLaPe02}), eq.~\eqref{.} amounts to
\begin{equation}\label{**}
\sum_{n\geq 0} \ip{\eta_n \otimes \phi^n_q}{\eta_{p-q+n} \otimes \phi^{p-q+n}_p} = 1 \quad \textrm{for all } p\geq q\geq 0 .
\end{equation}
Since
$$
\sum_{n\geq 0} \no{\eta_{j+n} \otimes \phi^{j+n}_k}^2 = \sum_{n\geq 0} \no{\phi^{j+n}_k}^2 \leq 1 \quad \textrm{for all } j,k\in\N ,
$$
by Cauchy-Schwartz inequality Eq. \eqref{**} holds if and only if
$$
\eta_n \otimes \phi^n_q = \eta_{p-q+n} \otimes \phi^{p-q+n}_p \ \forall n\in\N \quad \textrm{and} \quad \sum_{n\geq 0} \no{\phi^{n}_q}^2 = \sum_{n\geq 0} \no{\phi^{p-q+n}_p}^2 = 1 .
$$
for all $p\geq q \geq 0$. These two conditions are in turn equivalent to the following
\begin{enumerate}
\item $\phi_p^n = 0$ for all $0\leq n < p$, and $\sum_{n\geq p} \no{\phi_p^n}^2 = 1$;
\item $\eta_n \otimes \phi^n_0 = \eta_{p+n} \otimes \phi^{p+n}_p$ for all $n,p\in\N$.
\end{enumerate}

If conditions (1) and (2) hold, let $n_0$ be such that $\phi_0^{n_0} \neq 0$. Condition (2) with $n=n_0$ then implies $\eta_{p+n_0} = \lambda_p \eta_{n_0}$ for all $p\in \N$, and thus $c_{n_0 + p , n_0 + q} = \overline{\lambda}_p \lambda_q$ for all $p,q\geq 0$.

If conversely there exists $n_0\in\N$ and a sequence $\{ \lambda_n \}_{n\geq n_0}$ such that $c_{m,n} = \overline{\lambda}_m \lambda_n$ for all $m,n\geq n_0$, then by Cauchy-Schwartz inequality $\eta_n = \overline{\lambda}_{n_0} \lambda_n \eta_{n_0}$  for all $n\geq n_0$. Choose $\phi_k$ in Eq.~\eqref{def. phik} such that
$$
\phi_q^n = \delta_{n_0 , n-q} \overline{\lambda}_n \lambda_{n_0} \phi \quad \forall n\in\N ,
$$
with $\phi\in \hh$ such that $\no{\phi} = 1$ (in the above formula, $\phi_q^n = 0$ if $n<n_0$). Then the sequence $\{\phi_q^n \}_{q,n\in\N}$ satisfies conditions (1) and (2) above.
\end{proof}

A phase observable can be preprocessing clean without being extremal. Indeed, we have the following example.

\begin{example}\rm
Let $n_0\ge1$. Define $c_{nm}=1$ if $n\ge n_0$ and $m\ge n_0$, and $c_{nm}=\delta_{nm}$ (Kronecker delta) otherwise.
Obviously, $(c_{nm})$ is a phase matrix and the corresponding phase observable $\E$ is not unitarily equivalent with $\Ecan$.
As a consequence of Proposition \ref{PropEcan} we conclude that $\E\dsim\Ecan$. However, in Proposition \ref{prop:3}  we have shown that $\E\sim_U\Ecan$ equals $\E\csim\Ecan$. Therefore, $\E$ and $\Ecan$ are not postprocessing equivalent.
Note that $c_{m,n} \in \R$ for all $m,n$ and the rank of $(c_{nm})$ is $n_0+1$ (this can be seen easily by noting that the $\eta$-sequence of $(c_{nm})$ can be chosen such that the first $n_0+1$ $\eta_n$'s are mutually orthonormal and $\eta_n=\eta_{n_0+1}$ for all $n> n_0+1$), hence $\E$ is not extremal by Prop.~\ref{prop:creale}. Moreover, $\E$ is approximately sharp by Prop.~\ref{prop:phase-u} and, thus, postprocessing clean.
\end{example}

We have seen that $\Ecan$ is optimal in all four ways. However, there are also other phase observables sharing this feature. This is illustrated in the following example.

\begin{example}
Let $\{ f_1 , f_2 \}$ be the canonical basis of $\C^2$, and let $\eta_0 = f_1$, $\eta_1 = f_2$, $\eta_2 = 2^{-1/2} (f_1 + f_2)$, $\eta_3 = 2^{-1/2} (f_1 + i f_2)$, and $\eta_n = f_1$ for all $n\geq 4$. Define the phase matrix $c_{m,n} = \ip{\eta_m}{\eta_n}$, and let $\E$ be the associated phase observable.

It has been shown in \cite{KiPe08} that $(c_{m,n})$ is a rank $2$ phase matrix which is extremal in the convex set of phase matrices, hence $\E$ is extremal in the convex set $\obs$. Moreover, $\lim_{m\to \infty} c_{m,m+k} = 1$ for all $k\in\N$, and therefore $\E$ is approximately sharp by Proposition \ref{prop:phase-u}.  It follows from Proposition \ref{prop:3} that it is also postprocessing clean. By Proposition \ref{PropEcan}  the observable $\E$ is preprocessing clean. However, $\E \not\sim_U \Ecan$ since the ranks of the associated phase matrices are different. 
\end{example}

\noindent{\bf Acknowledgments.} The authors thank Jukka Kiukas and Kari Ylinen for valuable discussions.

\end{document}